\documentclass[12pt]{article}
\usepackage{amsmath,bbm,amssymb}
\usepackage{amsthm,amscd}
\usepackage{graphicx, graphics}
\usepackage[labelsep=period]{caption}
\usepackage{comment}

\newtheorem{proposition}{Proposition}

\usepackage{amsmath,bbm,amssymb}
\usepackage{amsthm,amscd}
\usepackage{graphicx, graphics}
\usepackage{multirow}

\newcommand{\bm}[1]{\mbox{\boldmath$ #1 $\unboldmath}}

\newcounter{example}
\newenvironment{example}[1][]{\refstepcounter{example}\par\medskip\noindent%
   \textbf{Example~\theexample. #1}\rmfamily}{\medskip}

\def\OA{{\mbox{\small OA}}}
\def\NOLH{{\mbox{\small NOLHD}}}
\newcommand{\bX}{\textbf{X}}
\newcommand{\bx}{\textbf{x}}
\newcommand{\bt}{\textbf{t}}
\newcommand{\bD}{\textbf{D}}
\newcommand{\bA}{\textbf{A}}
\newcommand{\bB}{\textbf{B}}

\newcommand{\bC}{\textbf{C}}
\newcommand{\bL}{\textbf{L}}
\newcommand{\bM}{\textbf{M}}
\newcommand{\bV}{\textbf{V}}
\newcommand{\bI}{\textbf{I}}

\newcommand{\bZ}{\textbf{Z}}
\newcommand{\brho}{\mbox{\boldmath${\rho}$} }

\begin{document}
\baselineskip=22pt
\vskip 20pt

\begin{center}
{\Large \bf Efficient Experimental Design for Regularized Linear Models}
\\

\vskip 15pt
\begin{minipage}[c]{8cm}
\centering C. Devon Lin \\
Department of Mathematics and Statistics\\
Queen's University, Canada
\end{minipage}
\hspace{0.5cm}
\begin{minipage}[c]{6cm}
\centering Peter Chien\\
Department of Statistics\\
University of Wisconsin-Madison
\end{minipage}
\vskip 15pt
\begin{minipage}[c]{6cm}
\centering Xinwei Deng\\
Department of Statistics\\
Virginia Tech
\end{minipage}
\vskip 20pt
\end{center}

\begin{abstract}

Regularized linear models, such as Lasso, have attracted great attention in statistical learning and data science.
However, there is sporadic work on constructing efficient data collection for regularized linear models.
In this work,
we propose an experimental design approach, using nearly orthogonal Latin hypercube designs, to enhance the variable selection accuracy of the regularized linear models.
Systematic methods for constructing such designs are presented.
The effectiveness of the proposed method is illustrated with several examples.
\end{abstract}


\section{Introduction}\label{sec: intro}

In statistical learning and data sciences,
regularized linear models have attracted great attention across multiple disciplines (Fan, Li, and Li, 2005; Hesterberg et al., 2008; Huang, Breheny, and Ma, 2012; Heinze, Wallisch, and Dunkler, 2018).
Among various regularized linear models, the Lasso is one of the most well-known techniques on the $L_{1}$ regularization to achieve accurate prediction with variable selection (Tibshirani, 1996).
Statistical properties and various extensions of this method have been actively studied in recent years (Tibshirani, 2011; Zhao and Yu 2006; Zhao et al., 2019; Zou and Hastie, 2015; Zou, 2016).
However, there is sporadic work on constructing efficient data collection for regularized linear models.
In this article, we study the data collection for the regularized linear model from an experimental design perspective.

First, we give a brief description of the Lasso procedure. Consider a linear model
\begin{align}\label{eq: regression}
y = \bx^{T} \bm \beta + \epsilon,
\end{align}
where $\bx  = (x_{1}, \ldots, x_{p})^{T}$ is the vector of $p$ continuous predictor variables,
$y$ is the response value, $\bm \beta = (\beta_{1}, \ldots, \beta_{p})^{T}$ are the vector of
regression parameters, and the error term  $\epsilon$ is normally distributed with mean zero and variance $\sigma^{2}$.
Throughout, assume data are centered so that the model in \eqref{eq: regression}
has no intercept.
Suppose this model has a \textit{sparse} structure
for which only $p_{0}$ predictor variables are \textit{active} with non-zero regression coefficients,
where $p_{0} < p$.  Let $\mathcal{A}( \bm \beta ) = \{ j:  \beta_{j} \ne 0, j = 1, \ldots, p \}$ be the set of the indices of the active variables.
Then the cardinality of the set
$\mathcal{A}( \bm \beta )$
is $p_{0}$.

For a given $n \times p$ regression matrix
$\bX = (\bm x_{1}, \ldots, \bm x_{n})^{T}$,
and a given response vector $\bm y = (y_{1}, \ldots, y_{n}) ^{T}$,
the Lasso solution is
\begin{align}\label{eq: Lasso}
\hat{\bm \beta} = \mbox{arg}\min_{\bm \beta} [(\bm y - \bX \bm \beta)^{T}(\bm y - \bX \bm \beta) + \lambda \|\bm \beta\|_{l_{1}}],
\end{align}
where $\| \bm \beta\|_{l_1} = \sum_{i = 1}^{p} |\beta_{i}|$ and $\lambda$ is a tuning parameter.
Because the $l_1$ norm $\| \cdot \|_{l_{1}}$ is singular at the origin,
a desirable property of the Lasso is that some
coefficients of $\hat{\bm \beta}$ are exactly zero.
Then $\mathcal{A}(\bm \beta)$ can be estimated by
$\mathcal{A}(\hat{\bm \beta}) = \{j: \hat{\beta}_{j} \ne 0, j = 1, \ldots, p \}$. The number of false selections of the Lasso is
\begin{align}\label{eq: false_selection}
\gamma  = \# \{ j: j \in  \mathcal{A}(\hat{\bm \beta}) \mbox{ but } j \notin \mathcal{A}( \bm \beta ) \}
+ \# \{ j: j \notin  \mathcal{A}(\hat{\bm \beta}) \mbox{ but } j \in \mathcal{A}( \bm \beta ) \},
\end{align}
where $\#$ denotes the set cardinality, the first term counts the number of false positives and the second term counts the number of false negatives.

The scope of this work is in developing experimental design techniques to construct the regression matrix $\mathbf{X}$ in \eqref{eq: Lasso} in order to minimize the value of $\gamma$ in (\ref{eq: false_selection}), the number of false selection.
Based on the probability properties of regularized linear models, it often requires large randomness in the regression matrix (Jung et al. 2019).
Thus, a straightforward way is to take $\bX$ to be an independently and identically distributed (i.i.d.) sample.
However, from an experimental design perspective, the points of an i.i.d.~sample is not well stratified in the design space
(Box, Hunter, and Hunter, 2005; Wu and Hamada, 2009). To improve upon this scheme, we propose to take $\bX$ to be a nearly orthogonal Latin hypercube design (NOLHD)
that is a Latin hypercube design with nearly orthogonal columns.
A Latin hypercube design is a space-filling design that
achieves maximum uniformity when the points are projected onto
any one dimension  (McKay, Beckman, and Conover, 1979). An NOLHD simultaneously possesses two desirable properties: low-dimensional stratification
and nearly orthogonality. Owen (1992)
stated the advantage of using Latin hypercube designs for
fitting additive models. It was discussed in
Section 3 of Owen (1992) that the least-squares estimates of
the regression coefficients of an additive regression with a Latin hypercube design can
have significant smaller variability than their
counterparts under an i.i.d. sample.
Since the model in \eqref{eq: regression}
is additive, $\hat{\bm \beta}$ in \eqref{eq: Lasso} associated with a Latin hypercube design
is expected to be superior to that with an i.i.d. sample.
Both random Latin hypercube designs and NOLHDs are popular in computer experiments (Lin and Tang,  2015).
It is advantageous to
use NOLHDs instead of
random Latin hypercube designs for the Lasso problem because the former have guaranteed small columnwise correlations.
When the regression matrix $\bX$ is taken to be an NOLHD,
its small columnwise correlations allow the active variables less correlated with the inactive variables,
thus improving the selection accuracy of the Lasso.

There is some consistency between
the concept of NOLHDs and the sparsity concept in variable selection.
The sparsity assumption we have made earlier
for the model in (\ref{eq: regression}) states that
only $p_0$ variables in the model are active and does not specify
which $p_0$ variables are active. If the regression matrix $\bX$ for this model
is an  NOLHD of $n$ runs for $p$ input variables,
when the points of this design are projected onto any $p_0$ dimensions, the resulting design
still retains the NOLHD structure for the $p_0$ factors.
Note that an NOLHD has more than two levels
and spreads the points evenly in the design space, not restricted to the boundaries only.
Since the number of false selections $\gamma$ in \eqref{eq: false_selection} has a nonlinear relation with the regression matrix $\bX$,
the use of an NOLHD for the Lasso problem is more appropriate than a two-level design to exploit this complicated relation between $\gamma$ and $\bX$.

The remainder of the article is organized as follows.
Section \ref{sec: OLHD-construction} introduces
a new criterion to measure NOLHDs, and presents two systematic methods for constructing
such designs for the Lasso problem. Section \ref{sec: Simulation} provides numerical examples to bear out the effectiveness of the proposed method.
 The numerical examples in Section
\ref{sec: Simulation} clearly indicate the superiority of NOLHDs
over two-level designs for the Lasso problem.
We provide a brief discussion in Section~\ref{sec: Discussion}.

\section{Methodology} \label{sec: OLHD-construction}

In this section we discuss the construction  of NOLHDs and how to use them for the Lasso problem.
We prefer NOLHDs
over random Latin hypercube designs because the latter are not guaranteed to have small columnwise correlations.
An an illustration, let $n = 64$ and $p=192$
and compute $\gamma$ in \eqref{eq: false_selection} for two different choices of $\bX$ in \eqref{eq: Lasso}.
Assume the model in \eqref{eq: regression} has $\sigma = 8$ and $\bm \beta = (0.05,0.2, \ldots, 3.0, 0 \ldots, 0)^{T}$,
where only the first $20$ coefficients are nonzeros, and the predictor variables take values on the hypercube $[-(64-1)/2, (64-1)/2]^{p}$.
The first method takes the design matrix $\bX$ to be a random Latin hypercube design constructed by \eqref{eq: lh}.
\begin{figure}[h]
\centering \resizebox{300pt}{250pt}
{\includegraphics{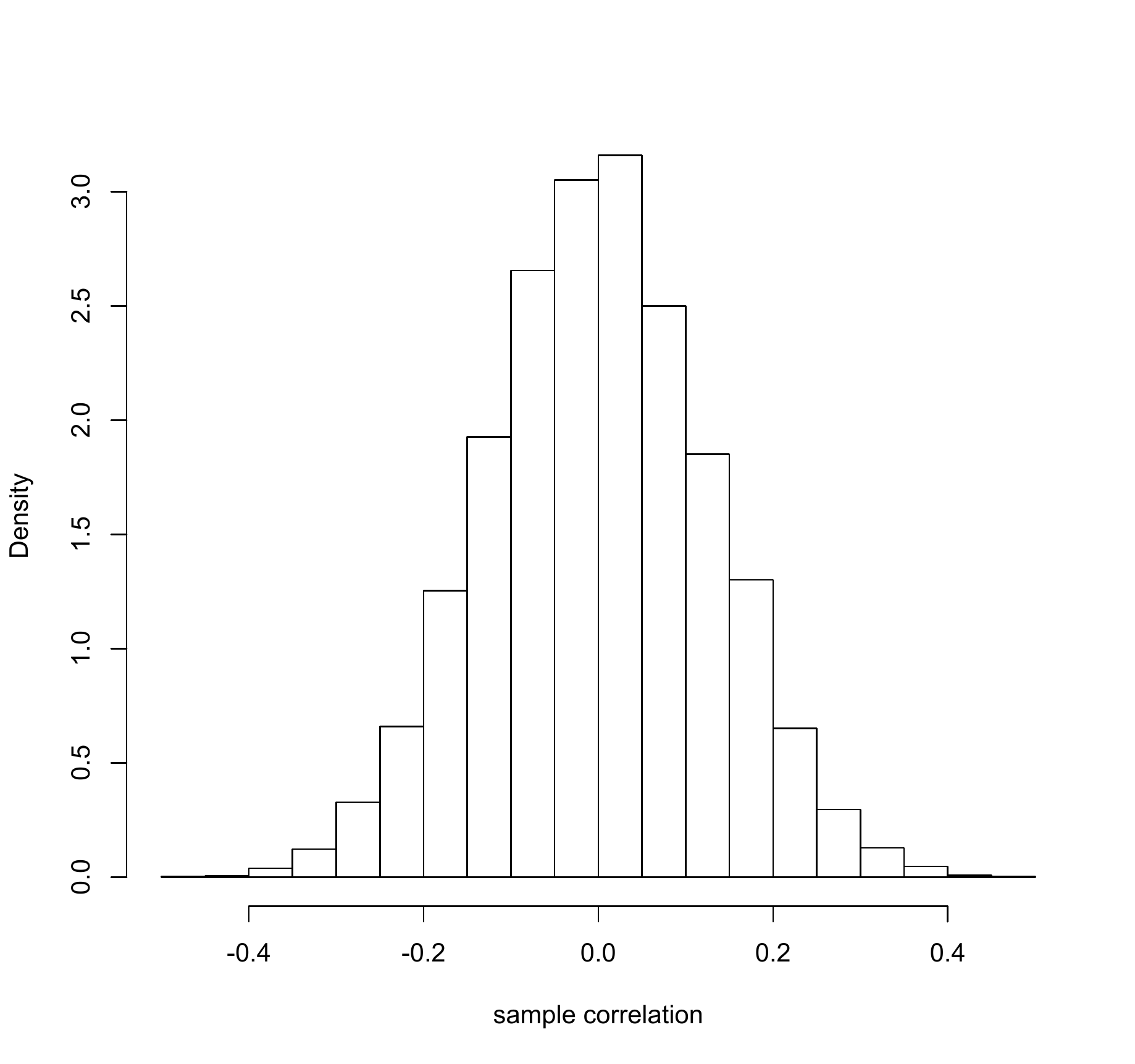}} \caption{Histogram of the sample correlations of a $64 \times 192$ random Latin hypercube design.}\label{fig: CorrPlot-lhd-hist}
\end{figure}
Fig.~\ref{fig: CorrPlot-lhd-hist} depicts the histogram of the columnwise sample correlations of one random Latin hypercube design,
where $21\%$ of the columnwise correlations
of the matrix are larger than $0.1$ in absolute values.
This method gives $\gamma = 36$. The second method takes the design matrix to be a $64 \times 192$ NOLHD from Section \ref{sec: construction-oa}, where 
the columnwise correlations of the matrix are very small. The second method gives $\gamma = 20$. The difference of $\gamma$ values of the two methods indicates that
the Lasso solution with an NOLHD can be far more
superior.

Here are some useful notation and definitions for constructing NOLHDs. The Kronecker product of an  $n \times p$ matrix $\bA=(a_{ij})$ and an $m \times q$ matrix $\bB=(b_{ij})$ is
\begin{center}
$\bA \otimes \bB =\left [\renewcommand{\arraystretch}{0.9}
\begin{array}{cccc}
  a_{11}\bB   &  a_{12}\bB   &   \ldots  &   a_{1p}\bB  \\
  a_{21}\bB   &  a_{22}\bB   &   \ldots  &   a_{2p}\bB  \\
  \vdots    &  \vdots    &   \ddots  &   \vdots   \\
  a_{n1}\bB &  a_{n2}\bB &   \ldots  &   a_{np}\bB
\end{array}\right],$
\end{center}
\noindent where $a_{ij}\bB$ is an $m \times q$ matrix whose ($k,l$) entry is $a_{ij}b_{kl}$.
The correlation matrix of an $n\times p$ matrix $\bX=(x_{ij})$  is
\begin{equation}\label{eq:rhomat}
\brho = \left(
\begin{array}{rrrr}
\rho_{11} & \rho_{12} & \ldots & \rho_{1p}\\
\rho_{21} & \rho_{22} & \ldots & \rho_{2p}\\
\vdots & \vdots  &\ddots &\vdots\\
\rho_{p1} & \rho_{p2} & \ldots & \rho_{pp}\\
\end{array}
\right),
\end{equation}
\noindent where
\begin{equation}\label{eq:rho}
\rho_{ij}=\frac{\sum_{k=1}^n(x_{ki}-\bar{x}_i)(x_{kj}-\bar{x}_j)}{\sqrt{\sum(x_{ki}-\bar{x}_i)^2
\sum(x_{kj}-\bar{x}_j)^2}},
\end{equation}
\noindent represents the correlation between the $i$th and $j$th columns of $\bX$, $\bar{x}_i=n^{-1}\sum_{k=1}^n x_{ki}$ and $\bar{x}_j=n^{-1}\sum_{k=1}^n x_{kj}$. The matrix $\bX$ is  orthogonal if $\brho$ in (\ref{eq:rhomat}) is an identity matrix.

Let $\bD=(d_{ij})$ be an $n \times p$  random Latin hypercube in which each column is a random permutation of $1, \ldots, n$, and all columns are generated independently. Using $\bD$,  a random Latin hypercube design $\bZ=(z_{ij})$ on $[0,1]^p$ is generated through
\begin{align}\label{eq:lh1}
z_{ij}=\frac{d_{ij}-u_{ij}}{n},i=1,\ldots,n; j=1,\ldots,p,
\end{align}
where the $u_{ij}$'s are independent  uniform random variables on [0,1), and the $d_{ij}$'s and the $u_{ij}$'s are mutually independent.
If $\bZ$ needs to be defined on $[a,b]^p$ for general $a<b$, rescale $z_{ij}$ in (\ref{eq:lh1}) as
\begin{align}\label{eq: lh}
z_{ij} \leftarrow (b-a)z_{ij}+a.
\end{align}

We use $\NOLH(n, p)$ to denote an  $n \times p$  NOLHD.
For a pre-specified vector $\bt=(t_1,\ldots,t_q)$ with $0 \leq t_q \leq \cdots \leq t_1 \leq 1$,
the orthogonality of an NOLHD $\bX$ can be assessed by using
the {\em proportion correlation vector} given by
\begin{equation}\label{eq:delta}
\delta_{\bt}(\bX)= (\delta_{t_1}(\bX),\ldots,\delta_{t_q}(\bX)),
\end{equation}
\noindent where  $\delta_{t_k}(\bX)=\{p(p-1)\}^{-1}\sum_{i=1}^p\sum_{j \neq i}I(|\rho_{ij}|\leq t_k)$, $k=1,\ldots,q$, and $I(\cdot)$ is an indicator function.
For $k=1,\ldots,q$, this criterion computes the proportion of the $|\rho_{ij}|$'s not exceeding $t_k$. For two designs $\bX_1$ and $\bX_2$, $\bX_1$ is preferred over $\bX_2$ if
$\delta_{t_k}(\bX_1) > \delta_{t_k}(\bX_2)$ for $t_1,\ldots, t_q$.
For the Lasso problem, this new criterion has more
discriminating power
than the {\em maximum correlation} $\rho_m$ and {\em root average squared correlation} $\rho_{ave}$
criteria proposed
in Bingham, Sitter, and Tang (2009), where $\rho_m(\bX)=\hbox{max}_{i,j}|\rho_{ij}|$ and
and $\rho_{ave}(\bX)=\{\sum_{i<j}\rho_{ij}^2/[p(p-1)/2]\}^{1/2}$. Designs with similar values of $\rho_{m}$ and $\rho_{ave}$ may have different values of $\delta_{\bt}$.
For illustration, compare a randomly generated $64 \times 192$ i.i.d.~sample with an $\NOLH(64, 192)$ from Section \ref{sec: construction-oa}.
The former has $\rho_{ave}=0.124$ and $\rho_m=0.493$ and the latter has $\rho_{ave}=0.112$ and  $\rho_m=0.786$.
The two designs are indistinguishable in terms of $\rho_{ave}$.
But for $\bt=(0.1,0.05,0.01,0.005)$,
$\delta_{\bt}=(0.562,0.305,0.064,0.033)$ for the i.i.d. sample
and $\delta_{\bt}=(0.906,0.894,0.883,0.883)$ for the NOLHD, clearly indicating the superiority of the latter.

Sections \ref{sec: construction-oa} and \ref{sec: construction-gkp} present two systematic methods for constructing NOLHDs.
The first method was proposed by Lin, Mukerjee, and Tang (2009) and the second method
is a generalization of the method in Lin et al.\ (2010).
To assist readers in machine learning who may not be familiar with NOLHDs, we describe
these two methods in a self-contained fashion.
These two methods are easy to implement. Other construction methods for (nearly) orthogonal Latin hypercube designs include Owen (1994), Tang (1998), Ye (1998),  Steinberg and Lin (2006), Pang, Liu, and Lin (2009), and Sun, Liu, and Lin (2009, 2010), among others. However, they have run-size constraints and thus we do not consider here.
In the two constructions we will present,
an NOLHD with $n$ runs is obtained
from a Latin hypercube
in which the $n$ levels in each column are {\small $\{-(n-1)/2, \ldots, 0, \ldots,
(n-1)/2\}$} if $n$ is odd and {\small $\{-(n-1)/2,\ldots,-1/2,1/2,\ldots,(n-1)/2\}$} if $n$ is even.

\subsection{A Construction Method Using Orthogonal Arrays}\label{sec: construction-oa}

Lin, Mukerjee, and Tang (2009) proposed a method for constructing nearly orthogonal Latin hypercubes using orthogonal arrays.
Recall that an orthogonal array $\OA(n,p,s)$ of strength two is an $n \times p$ matrix with levels $1,\ldots,s$ such that, for any two columns,  all level combinations appear equally often (Hedayat, Sloane, and Stufken, 1999).
Let $\bA$ be an $\OA(s^2, 2f, s)$ and let $\bB = (b_{ij})$ be an $s \times p$ Latin hypercube.
This method works as follows.

\noindent {Step 1.} For $j = 1, \ldots, p$, obtain an $s^2 \times (2f)$
matrix $\bA_j$ from $\bA$ by replacing the symbols $1,   \ldots, s$ in
the latter by $b_{1j},   \ldots, b_{sj}$, respectively, and
 partition $\bA_j$ to  $\bA_{j1}, \ldots, \bA_{jf}$,
each of two columns.

\noindent {Step 2.} Let \[
{
\renewcommand{\arraystretch}{0.8}
\bV = \left[
\begin{array}{rr}
1 & -s \\
s & 1 \\
\end{array}
\right]. }
\] \noindent For $j = 1, \ldots, p$, obtain an $s^2 \times (2f)$
matrix $$\bM_j = [\bA_{j1}\bV, \ldots, \bA_{jf}\bV].$$

\noindent {Step 3.} For $n=s^2$ and $q=2pf$, define an $n\times q$ matrix $\bM$ = [$\bM_{1}, \ldots,
\bM_{p}$].

Lemma 1 from Lin, Mukerjee, and Tang (2009) captures the structure of $\bM$.

{\bf \sc Lemma 1.} (a) {\sl The matrix $\bM$ constructed above is an $s^2 \times (2pf)$ Latin hypercube.} \\ (b) {\sl The correlation matrix of $\bM$ is  $\brho(\bM) =
\brho(\bB) \otimes
\bI_{2f}$, where $\bI_{2f}$ is the identity matrix of order $2f$.}

Observe that the proportion correlation $\delta_{t_k}$ in (\ref{eq:delta}) of  $\bM$ is
\begin{equation}\label{eq:deltaM}
\delta_{t_k}(\bM)=[ p(2f-1)+(p-1)\delta_{t_k}(\bB) ] / (2pf-1), \ \hbox{for} \  k=1,\ldots,q.
\end{equation}

\begin{example}{Example}\label{eg:exam2}
Let $\bA$ be an $\OA(49, 8, 7)$ from Hedayat, Sloane, and Stufken (1999) and let $\bB$ be an $\NOLH(7,12)$ given by
$${\footnotesize
\left(
\begin{array}{rrrrrrrrrrrr}
  -3&   0&  -1 &  0&   3 &  3&   0&  -2&   1&  -3 & -1&  -3 \\[-2pt]
  -2&  -1&   1 & -3&  -1 & -3&   1&  -3&  -2&  -1 &  1&   3 \\[-2pt]
  -1&   3&   0 &  3&   0 & -2&   2&   0&  -1&   3 & -3&  -1 \\[-2pt]
   0&  -2&   3 &  2&  -2 &  2&  -2&   1&  -3&   1 &  2&  -2 \\[-2pt]
   1&   1&  -3 & -1&  -3 &  1&  -3&  -1&   3&   2 &  0&   1 \\[-2pt]
   2&  -3&  -2 &  1&   1 & -1&   3&   2&   2&   0 &  3&   0 \\[-2pt]
   3&   2&   2 & -2&   2 &  0&  -1&   3&   0&  -2 & -2&   2 \\
\end{array}\right),}$$
\noindent where $\rho_{ave}(\bB)=0.3038$, $\rho_m(\bB)=0.9643$, and $\delta_{\bt}(\bB)=(\delta_{0.1},\delta_{0.05},\delta_{0.01},
\delta_{0.005})$=\\$(0.500,0.364,0.136,0.136)$. Here, the matrix $\bM$ from Lemma~1 is an $\NOLH(49,96)$ with $\rho_{ave}(\bM)=0.1034$ and
$\rho_m(\bM)=0.9643$. From (\ref{eq:deltaM}), $\delta_{\bt}(\bM)=(\delta_{0.1},\delta_{0.05},\delta_{0.01},\\
\delta_{0.005})$=$(0.942,0.926,0.9,0.9)$.
In general, if $\bB$ is an $\NOLH(7,p)$, Lemma~1 gives an $\NOLH(49,8p)$.
\end{example}

\subsection{A Construction Method Using the Kronecker Product}\label{sec: construction-gkp}

We now propose a generalization of the method in Lin et al.\ (2010) for constructing NOLHDs.
This generalization provides designs with better low-dimensional projection properties than those obtained
in Lin et al.\ (2010).

For $j=1,
\ldots, m_2$, let $\bC_j=(c_{ik}^j)$  be an $n_1 \times m_1$ Latin hypercube and let
$\bA_j=(a_{ik}^{j})$ be an $n_1\times m_1$ matrix with entries $\pm
1$.  Let $\bB=(b_{ij})_{n_2 \times m_2}$
be an $n_2 \times m_2$ Latin hypercube,
let $\bD=(d_{ij})_{n_2 \times m_2}$ be a matrix with entries $\pm 1$,
and let $r$ be a real number. Our proposed method  constructs\\[-0.3in]

\begin{eqnarray}\label{gene2}
 \bM   &=\left[
 \begin{array}{rrrr}
b_{11}\bA_1+r d_{11}\bC_1 & b_{12}\bA_2+r d_{12}\bC_2 & \ldots &  b_{1m_2}\bA_{m_2}+r d_{1m_2}\bC_{m_2}  \\
b_{21}\bA_1+r d_{21}\bC_1 & b_{22}\bA_2+r d_{22}\bC_2 & \ldots &  b_{2m_2}\bA_{m_2}+r d_{2m_2}\bC_{m_2} \\
\vdots             & \vdots             &\ddots & \vdots  \\
b_{n_21}\bA_1+r d_{n_21}\bC_1 &  b_{n_22}\bA_2+r d_{n_22}\bC_2 & \ldots &  b_{n_2m_2}\bA_{m_2}+r d_{n_2m_2}\bC_{m_2}\\
\end{array}
 \right].
\end{eqnarray}

In contrast, the method in Lin et al.\ (2010) constructs
\begin{equation}\label{eq1}
\bL=\bA\otimes \bB + r \bC\otimes \bD,
\end{equation}
\noindent where $\bA=(a_{ij})_{n_1 \times m_1}$ is a matrix with entries $\pm 1$,
$\bC=(c_{ij})_{n_1 \times m_1}$ is an $n_1 \times m_1$ Latin hypercube, and $\bB$, $\bD$ and $r$ are as in (\ref{gene2}).
Lin et al.\ (2010) provided the conditions for $\bL$ to be a nearly orthogonal Latin hypercube.  When projected onto some pairs of predictor variables, points in the design in (\ref{eq1}) lie on straight lines,
which may not be desirable for the Lasso problem. Such projection patterns are due to the use of the same $\bA$ and the same $\bC$ for each entry of $\bB$ and $\bD$ in (\ref{eq1}).
The generalization in (\ref{gene2}) uses different $\bA_j$'s and $\bC_j$'s to eliminate this undesirable projection pattern.
Proposition~1 establishes conditions for $\bM$ in (\ref{gene2}) to be a Latin hypercube.

\begin{proposition}
Let $r = n_2$. Then the
design $\bM$ in (\ref{gene2}) is a Latin hypercube if \\
one of the following two conditions holds: \\
\indent
(a) For $j=1,\ldots,m_2$, the $\bA_j$ and $\bC_j$ satisfy that for $i=1,\ldots,m_1$,
$c_{pi}^j = - c_{p'i}^j$ and $a_{pi}^j = a_{p'i}^j$ hold simultaneously. \\
\indent (b) For $k=1,\ldots,m_2$, the entries of $\bB$ and $\bD$ satisfy the condition that
$b_{qk} = - b_{q'k}$ and $d_{qk} = d_{q'k}$ hold simultaneously.
\end{proposition}

Proposition~1 can be verified by using an argument similar to
the proof of Lemma~1 in Lin et al.\ (2010) and thus is omitted.  Proposition~2 studies the orthogonality of $\bM$ in terms of $\bA_j$'s, $\bB$, $\bC_j$'s and $\bD$.

\begin{proposition}
Suppose $\bA_j$'s, $\bB$, $\bC_j$'s, $\bD$ and $r$ in (\ref{gene2}) satisfy  condition (a) or (b) in Proposition~1 and $\bM$ in (\ref{gene2}) is a Latin hypercube.
In addition, assume that $\bA_j$s, $\bB$, and $\bD$ are orthogonal and
that $\bB^T\bD = 0$ or $\bA_j^T\bC_j = 0$ holds for all $j$s. Then we have that\\
(a) $\rho_m(\bM)=\hbox{Max}\{w_1\rho_m(\bC_{j}),
  j=1, \ldots, m_2\}$, where $w_1= n_2^2(n_1^2-1)/(n_1^2n_2^2-1)$.\\
(b) $\rho_{ave}(\bM)=\sqrt{w_2\sum_{j=1}^{m_2}\rho_{ave}^2(\bC_{j})/m_2}$,
  where
  $w_2=(m_1-1)w_1^2/(m_1m_2-1)]$. \\
(c) $\delta_{t_k}(\bM) \geq \sum_{j=1}^{m_2}\delta_{t_k}(\bC_j)/m_2$ for $k=1,\ldots,q$.\\
 (d) The matrix $\bM$ is orthogonal if and only if $\bC_{1}, \ldots, \bC_{m_2}$
  are all orthogonal.
\end{proposition}

\begin{proof}
Let $\bM_{jk}$ and $\bM_{j'k'}$ be the $[(j-1)m_2+k]$th and $[(j'-1)m_2+k']$th columns of $\bM$ in (\ref{gene2}), respectively. Take $n=n_1n_2$. Let
$\rho(\bM_{jk},\bM_{j'k'})$ be the correlation between  $\bM_{jk}$ and
 $\bM_{j'k'}$  defined in (\ref{eq:rho}). Express ${\footnotesize 12^{-1}n(n^2-1)\rho(\bM_{jk},\bM_{j'k'})}$ as $$\sum_{i_1=1}^{n_2}
\sum_{i_2=1}^{n_1}(b_{i_1j}a_{i_2k}^j+n_2d_{i_1j}c_{i_2k}^j)
(b_{i_1j'}a_{i_2k'}^{j'}+n_2d_{i_1j'}c_{i_2k'}^{j'}),$$ \noindent which equals
{\footnotesize
\begin{eqnarray*}
&&
\sum_{i_1=1}^{n_2}b_{i_1j}b_{i_1j'}\sum_{i_2=1}^{n_1}a_{i_2k}^ja_{i_2k'}^{j'}+
n_2\sum_{i_1=1}^{n_2}d_{i_1j}b_{i_1j'}\sum_{i_2=1}^{n_1}c_{i_2k}^ja_{i_2k'}^{j'}\\
& & \quad
+n_2\sum_{i_1=1}^{n_2}b_{i_1j}d_{i_1j'}\sum_{i_2=1}^{n_1}a_{i_2k}^jc_{i_2k'}^{j'}+
n_2^2\sum_{i_1=1}^{n_2}d_{i_1j}d_{i_1j'}\sum_{i_2=1}^{n_1}c_{i_2k}^jc_{i_2k'}^{j'}\\
&=&
\sum_{i_1=1}^{n_2}b_{i_1j}b_{i_1j'}\sum_{i_2=1}^{n_1}a_{i_2k}^ja_{i_2k'}^{j'}
+n_2^2\sum_{i_1=1}^{n_2}d_{i_1j}d_{i_1j'}\sum_{i_2=1}^{n_1}c_{i_2k}^jc_{i_2k'}^{j'}.\\
\end{eqnarray*}}
\noindent Thus, $\rho(\bM_{jk},\bM_{j'k'})$ is zero for $j \neq j'$ and is
 $ n_2^2(n_1^2-1)\rho_{kk'}(\bC_j)/(n^2-1)$ for $j=j'$ and $k\neq k'$.
 By the definitions of $\rho_m$ and $\rho_{ave}$, the results in (a) and (b) hold.  Note that for $k=1,\ldots,q$,
\begin{eqnarray*}
\delta_{t_k}(\bM)&=&\{m_2(m_2-1)m_1^2+m_1(m_1-1)\sum_{j=1}^{m_2}\delta_{t_k}(\bC_j)\}/\{m_1m_2(m_1m_2-1)\}\\
&=&\sum_{j=1}^{m_2}\delta_{t_k}(\bC_j)/m_2+[(m_2-1)m_1\{1- \sum_{j=1}^{m_2}\delta_{t_k}(\bC_j)/m_2\}]/(m_1m_2-1).
 \end{eqnarray*}
The result in (c) now follows because $\sum_{j=1}^{m_2}\delta_{t_k}(\bC_j)/m_2\leq 1$.
By (a),
(b) and (c), (d) is evident. This completes the proof.
\end{proof}

Proposition~2 expresses the near orthogonality of $\bM$ in
(\ref{gene2}) in terms of that of $\bC_j$'s and establishes conditions for $\bA_j$, $\bB$ and $\bD$ in order for $\bM$ to be an orthogonal Latin hypercube.
The required matrices in Proposition~2 can be chosen as follows.
First, orthogonal matrices $\bA_j$'s and $\bD$  are readily available from Hadamard matrices when $n_1$ and $n_2$ are multiples of four.
Second, orthogonal Latin hypercubes $\bB$ are available from Pang, Liu, and Lin (2009), Lin, Mukerjee, and Tang (2009), Lin et al.\ (2010),  among others.
If $\bA_j$, $\bB$ and $\bD$ are orthogonal, and either $\bB^T\bD=0$ or $\bA_j^T\bC_j=0$, then $\bM$ is orthogonal when $\bC_j$'s are orthogonal Latin hypercubes.
If $\bC_j$'s are NOLHDs like those from  Lin, Mukerjee, and Tang (2009) and Lin et al.\ (2010), then $\bM$ is nearly orthogonal.
If $\bC_1$ is an NOLHD, $\bC_2, \ldots, \bC_{m_{2}}$ can be obtained by permuting the rows of $\bC_1$.

\begin{example}{Example}\label{eg:exam1}
Let
{\footnotesize $$ \bB=\frac{1}{2}\left(
\begin{array}{rrrr}
 1 & -3&  7&  5    \\[-2pt]
 3 &  1&  5& -7    \\[-2pt]
 5 & -7& -3& -1    \\[-2pt]
 7 &  5& -1&  3    \\[-2pt]
-1 &  3& -7& -5    \\[-2pt]
-3 & -1& -5&  7    \\[-2pt]
-5 &  7&  3&  1    \\[-2pt]
-7 & -5&  1& -3    \\
\end{array}\right),
\bD=\left(
\begin{array}{rrrr}
 1&  1&  1&  1   \\[-2pt]
 1&  1& -1& -1   \\[-2pt]
 1& -1&  1& -1   \\[-2pt]
 1& -1& -1&  1   \\[-2pt]
-1&  1&  1&  1   \\[-2pt]
-1&  1& -1& -1   \\[-2pt]
-1& -1&  1& -1   \\[-2pt]
-1& -1& -1&  1   \\
\end{array}\right),$$}
{\footnotesize
$$\bA_1=\left(
\begin{array}{rrrrrr}
    1 &    1&     1 &    1  &   1  &   1 \\[-2pt]
    1 &   -1&     1 &    1  &   1  &  -1 \\[-2pt]
   -1 &    1&    -1 &    1  &   1  &   1 \\[-2pt]
   -1 &   -1&     1 &   -1  &   1  &   1 \\[-2pt]
    1 &   -1&    -1 &    1  &  -1  &   1 \\[-2pt]
   -1 &    1&    -1 &   -1  &   1  &  -1 \\[-2pt]
   -1 &   -1&     1 &   -1  &  -1  &   1 \\[-2pt]
   -1 &   -1&    -1 &    1  &  -1  &  -1 \\[-2pt]
    1 &   -1&    -1 &   -1  &   1  &  -1 \\[-2pt]
    1 &    1&    -1 &   -1  &  -1  &   1 \\[-2pt]
    1 &    1&     1 &   -1  &  -1  &  -1 \\[-2pt]
   -1 &    1&     1 &    1  &  -1  &  -1 \\
\end{array}\right), \hbox{ and }
\bC_1=\frac{1}{2}\left(
\begin{array}{rrrrrr}
 -11 & -9 &  9&  11&   5 &  1   \\[-2pt]
  -9 &  5 & -1&  -5&  -9 & 11   \\[-2pt]
  -7 & 11 & -3&   3&   1 & -7   \\[-2pt]
  -5 & -1 & -9&  -9&  -1 & -9   \\[-2pt]
  -3 & -7 &  5& -11&   7 & -1   \\[-2pt]
  -1 &  9 & -7&   5&   9 &  5   \\[-2pt]
   1 & -3 &  7&  -7&  -7 &  3   \\[-2pt]
   3 &-11 &-11&   9& -11 & -3   \\[-2pt]
   5 &  7 & 11&   7&  -5 & -5   \\[-2pt]
   7 & -5 & -5&   1&  11 &  7   \\[-2pt]
   9 &  1 &  3&  -3&   3 &-11   \\[-2pt]
  11 &  3 &  1&  -1&  -3 &  9   \\
\end{array}\right).$$}
\noindent For $j=2,3,4$, obtain $\bA_j$ and $\bC_j$ by permuting the rows of $\bA_1$ and $\bC_1$, respectively. Using the above matrices,
$\bM$ in (\ref{gene2}) is a $96 \times 24$ orthogonal Latin hypercube.
\end{example}

\begin{example}{Example}\label{eg:exam3}
Let $\bC_1$ be an $\NOLH(25,24)$ constructed  by  Lemma~1
using an $\OA(25,6,5)$ from Hedayat, Sloane, and Stufken (1999) and an $\NOLH(5,4)$.
Permute the rows of $\bC_1$ to get an NOLHD $\bC_2$. Generate two $25 \times 24$ nearly orthogonal matrices, $\bA_1$ and $\bA_2$,
by using the Gendex DOE software associated with Nguyen (1996).
Using
$$
\bB=\left(
\begin{array}{rr}
\frac{1}{2}  & -\frac{1}{2} \\
-\frac{1}{2}  & \frac{1}{2}
\end{array}
\right)
\hbox{ and }
\bD=\left(\begin{array}{rr}
1  & 1\\
1 & 1
\end{array}
\right),
$$
\noindent $\bM$ in (\ref{gene2}) is an $\NOLH(50, 48)$.
\end{example}

\section{Numerical Illustration}\label{sec: Simulation}

In this section we provide numerical examples to compare the number of false selections $\gamma$ in \eqref{eq: false_selection}
with four different types of design matrices. Method I uses an NOLHD from Section \ref{sec: OLHD-construction}.
Method II uses a two-level design at levels $\pm (n-1)/2$. If $p>n-1$, a two-level design is often called a supersaturated design (Lin, 1993; Wu, 1993).
Method III uses a random Latin hypercube design (RLHD) constructed in \eqref{eq: lh}.
Method IV uses an i.i.d.~sample.
Denote by $\gamma_{NOLHD}$, $\gamma_{FD}$, $\gamma_{RLHD}$ and $\gamma_{IID}$ the $\gamma$ values of these methods, respectively.
Since the focus here is to compare the effect of the
regression matrix $\bX$ on the accuracy of the Lasso solution,  the response vector $\bm y$ from the model in \eqref{eq: regression} is generated with the same $\bm \epsilon = (\epsilon_{1}, \ldots, \epsilon_{n})^{T}$ for the four methods.
The tuning parameter $\lambda$ in \eqref{eq: Lasso} is selected by the five-fold cross-validation.
The package \textit{lars} (Efron, Hastie, Johnstone, and Tibshirani, 2003) in \textit{R} (R, 2010)
is used to compute the Lasso solution $\hat{\bm \beta}$ in \eqref{eq: Lasso}.
Examples below have different $p/n$ ratios.

\begin{example}{Example}\label{eg: example-2}
For the model in \eqref{eq: regression}, let $p=48$, $\sigma =8$, and $\bm \beta = (0.8,1.0, \ldots, 3, 0, \ldots, 0)^{T}$ with the last $36$ coefficients being zero.
Take $n = 50$ with $n \approx p$.
Method I takes the $\NOLH(50,48)$ in Example \ref{eg:exam3}. Method II uses a $50 \times 48$ nearly orthogonal two-level design from
the Gendex software based on the algorithm in Nguyen (1996).
Table \ref{tab: Fsel-Eg2} compares three quartiles of the $\gamma_{NOLHD}$, $\gamma_{FD}$, $\gamma_{RLHD}$ and $\gamma_{IID}$ values over 50 replications.
Fig. \ref{fig: boxplot-eg2} depicts the boxplots of $\gamma$ values of these methods.
Table \ref{tab: Fsel-Eg2} and Fig. \ref{fig: boxplot-eg2} clearly
indicate that $\gamma_{NOLHD}$ is smaller than $\gamma_{FD}$, $\gamma_{RLHD}$ and $\gamma_{IID}$.
\end{example}

\begin{table}[h]
\centering {\caption{Three quartiles of the $\gamma_{NOLHD}$, $\gamma_{FD}$, $\gamma_{RLHD}$ and $\gamma_{IID}$ values over 50 replications for Example \ref{eg: example-2}}\label{tab: Fsel-Eg2}}
\begin{tabular}{l|cccc}
\hline  &  NOLHD & FD & RLHD & IID \\
\hline
   median          &13.00         &18.00         &18.00             &20.00 \\
   1st quartile    &12.00         &14.00         &15.00             &16.00 \\
   3rd quartile    &15.00         &21.00          &23.00            &23.00  \\
 \hline
\end{tabular}
\end{table}

\begin{figure}[h!t]
\centering \resizebox{300pt}{250pt}
{\includegraphics{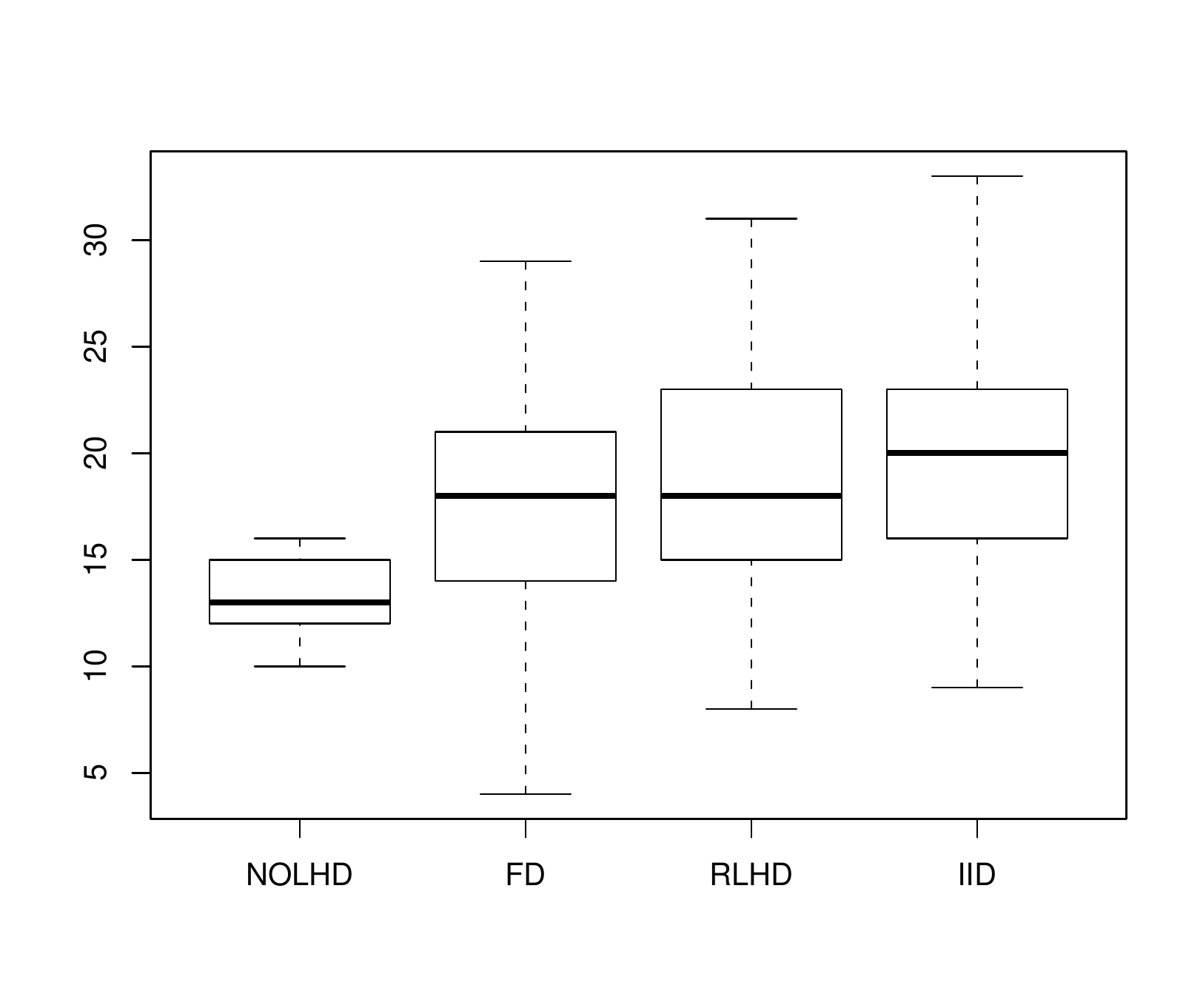}} \caption{Boxplots of the $\gamma_{NOLHD}$, $\gamma_{FD}$, $\gamma_{RLHD}$ and $\gamma_{IID}$ values over the 50 replications for Example \ref{eg: example-2}.}\label{fig: boxplot-eg2}
\end{figure}

\begin{example}{Example} \label{eg: example-3}
For the model in \eqref{eq: regression}, let $p=96$, $\sigma = 8$ and $\bm \beta = (0.2,0.4, \ldots, 3, 0, \ldots, 0)^{T}$
with the last $81$ coefficients being zero. Take $n=49$ with $p > n$.
Method I uses the $\NOLH(49,96)$ in Example \ref{eg:exam2}.  Method II uses an $E(s^2)$-optimal supersaturated design from the Gendex software associated with Nguyen (1996). Table \ref{tab: Fsel-Eg3} compares three quartiles of the $\gamma_{NOLHD}$, $\gamma_{FD}$, $\gamma_{RLHD}$ and $\gamma_{IID}$ values over 50 replications.
Fig. \ref{fig: boxplot-eg3} depicts the boxplots of $\gamma$ values of these methods.
Table \ref{tab: Fsel-Eg3} and Fig. \ref{fig: boxplot-eg3} show that $\gamma_{NOLHD}$, once more, significantly outperforms $\gamma_{FD}$, $\gamma_{RLHD}$ and $\gamma_{IID}$.
\end{example}

\begin{table}[h]
\centering {\caption{Three quartiles of the $\gamma_{NOLHD}$, $\gamma_{FD}$, $\gamma_{RLHD}$ and $\gamma_{IID}$ values over 50 replications for Example \ref{eg: example-3}.}\label{tab: Fsel-Eg3}}
\begin{tabular}{l|cccc}
\hline &  NOLHD & FD & RLHD & IID \\
\hline
  median         &17.50         &27.00           &25.00          &27.00   \\
  1st quartile    &15.00         &24.00         &22.25         &23.25  \\
  3rd quartile    &22.75         &30.00         &28.00         &29.00 \\
\hline
\end{tabular}
\end{table}

\begin{figure}[h!t]
\centering \resizebox{300pt}{250pt}
{\includegraphics{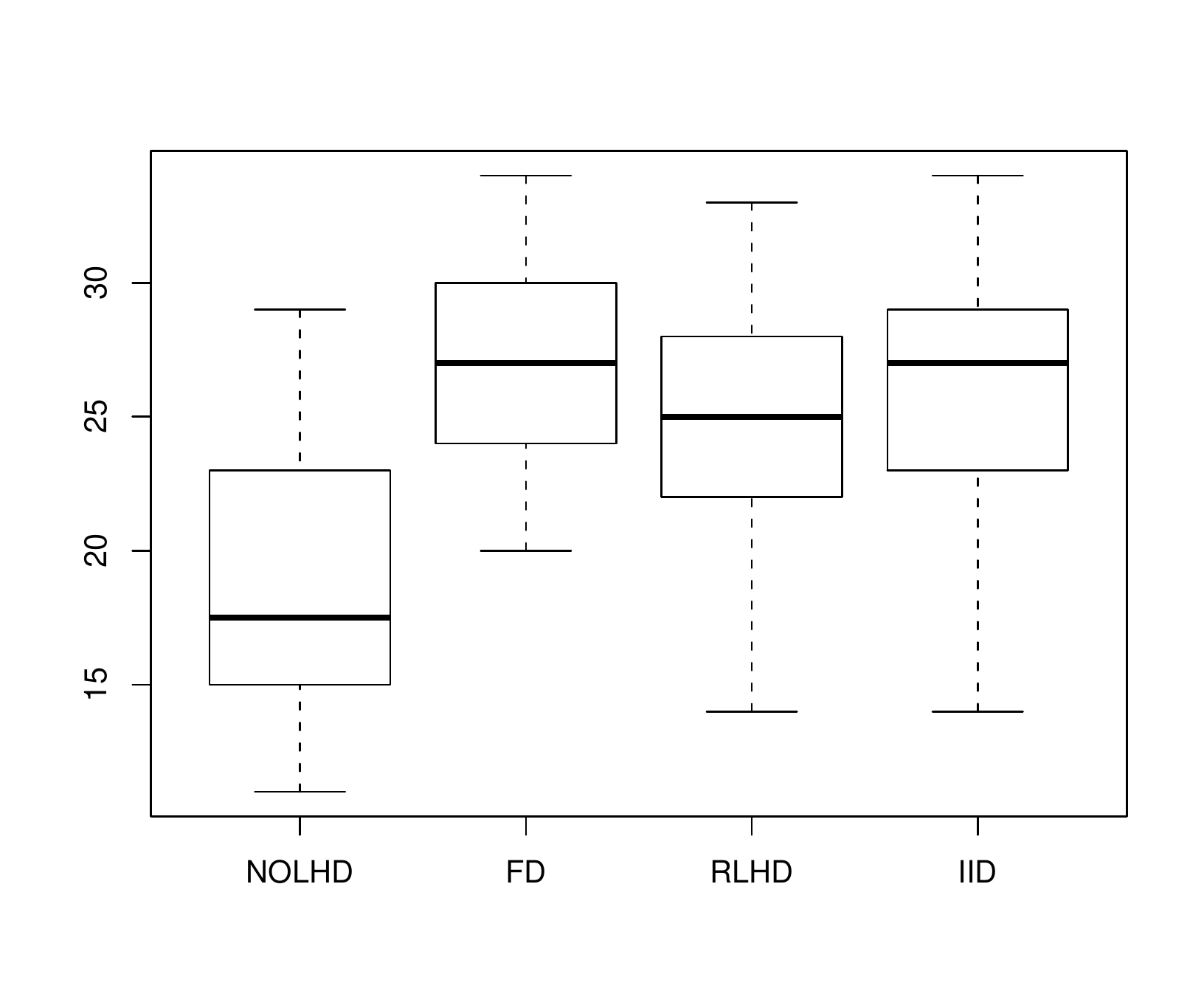}} \caption{Boxplots of the $\gamma_{NOLHD}$, $\gamma_{FD}$, $\gamma_{RLHD}$ and $\gamma_{IID}$ values over the 50 replications for Example \ref{eg: example-3}.}\label{fig: boxplot-eg3}
\end{figure}

\begin{example}{Example} \label{eg: example-4}
For the model in \eqref{eq: regression}, let $p=192$, $\sigma = 8$ and $\bm \beta = (0.05,0.2, \ldots, 3, 0, \ldots, 0)^{T}$
with the last $172$ coefficients being zero.
Take $n = 64$ with $p>n$. Method I uses an $\NOLH(64, 192)$ from Lemma~1 in Section \ref{sec: construction-oa}. Method II uses an $E(s^2)$-optimal supersaturated design from
the Gendex software associated with Nguyen (1996).
Table \ref{tab: Fsel-Eg4} compares three quartiles of the $\gamma_{NOLHD}$, $\gamma_{FD}$, $\gamma_{RLHD}$ and $\gamma_{IID}$ values over 50 replications.
Fig.  \ref{fig: boxplot-eg4} depicts the boxplots of $\gamma$ values for these methods, where $\gamma_{NOLHD}$ is much smaller than $\gamma_{FD}$, $\gamma_{RLHD}$ and $\gamma_{IID}$.
This example clearly demonstrates that the use of an NOLHD
leads to significant improvement of the Lasso solution.
 \end{example}

\begin{table}[h]
\centering {\caption{Three quartiles of the $\gamma_{NOLHD}$, $\gamma_{FD}$, $\gamma_{RLHD}$ and $\gamma_{IID}$ values over 50 replications in Example \ref{eg: example-4}. }\label{tab: Fsel-Eg4}}
\begin{tabular}{l|cccc}
\hline &  NOLHD &FD & RLHD & IID \\
\hline
  median         &27.00           &42.50       &43.00         &41.00  \\
  1st quartile   &23.00          &40.00        &34.25         &34.00 \\
  3rd quartile   &33.00          &46.00        &45.00         &45.00 \\
\hline
\end{tabular}
\end{table}

\begin{figure}[h!t]
\centering \resizebox{300pt}{250pt}
{\includegraphics{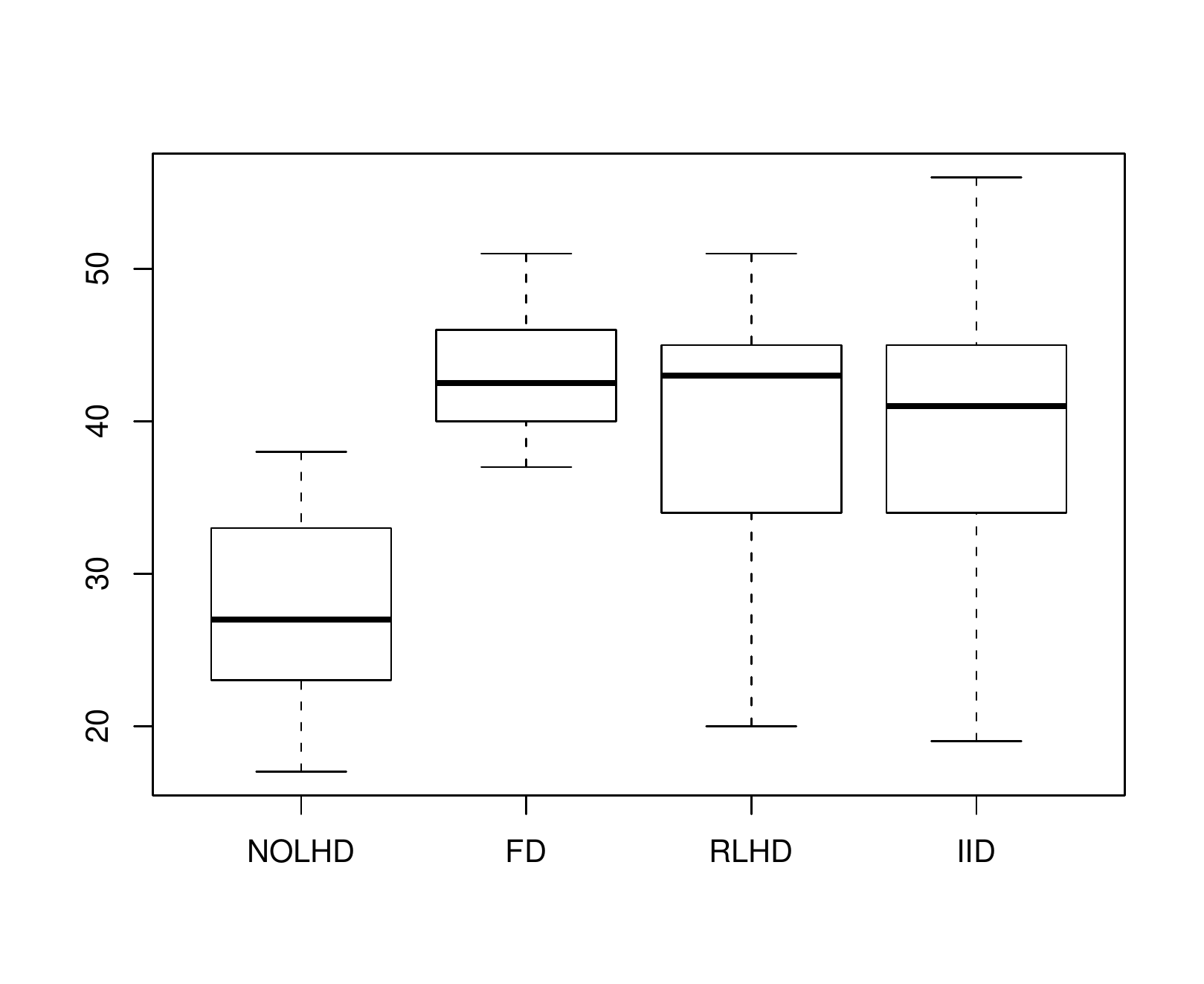}} \caption{Boxplots of the $\gamma_{NOLHD}$, $\gamma_{FD}$, $\gamma_{RLHD}$ and $\gamma_{IID}$ values over the 50 replications for Example \ref{eg: example-4}.}\label{fig: boxplot-eg4}
\end{figure}

These examples suggest that
the Lasso solution with an NOLHD is more accurate than
those of the competing designs.
Comparison of Fig. \ref{fig: boxplot-eg2}-- Fig. \ref{fig: boxplot-eg4} indicates that
the advantage of using NOLHDs in the Lasso problem
grows as the ratio $p/n$ increases.

\section{Discussion} \label{sec: Discussion}

We have proposed a method using NOLHDs from computer experiments to significantly enhance the variable selection accuracy of the Lasso procedure.
The effectiveness of this method has been successfully illustrated by several examples.
Design construction for the regularized linear models is a new research direction in design of experiments,
which can be applied in many areas, such comprehensive sensing (Song et al., 2016; Jung et al. 2019), and actuator placement (Du et al., 2019).
As an alternative to the proposed method, one may develop a model-based optimal design approach by extending the ideas of Meyer, Steinberg, and Box (1996) and Bingham and Chipman (2007).
Because  the Lasso solution in \eqref{eq: Lasso}
does not admit an analytic form, a potential difficulty in developing such an approach is to
introduce a sensible and computationally efficient criterion for the Lasso problem.
It will be of interest in a subsequent project
to study the proposed design strategy
for variants of the Lasso. A R package for the proposed method
is under development and will be released in the future.


\end{document}